\documentclass[12pt]{article}
\usepackage{authblk}
\usepackage{setspace} 

\usepackage[pdfencoding=auto,hidelinks]{hyperref}
\usepackage[inner=1.25in,outer=1in,bottom=1.1in, top=1.25in]{geometry}

\usepackage{graphicx}
\usepackage{subcaption}
\usepackage{float}
\usepackage{setspace}
\usepackage{placeins}
\usepackage{newtxtext,newtxmath} 
\usepackage{titlesec}     
\usepackage[titles]{tocloft} 

\usepackage{amsmath, amsfonts, amssymb}
\usepackage{geometry}
\usepackage{booktabs}
\usepackage[page]{appendix} 
\usepackage{bm}           
\usepackage{threeparttable}
\usepackage[table]{xcolor}
\usepackage{comment}
\usepackage{graphicx}
\usepackage{subcaption}
\usepackage{indentfirst}
\usepackage{multirow}
\usepackage{multicol}

\usepackage{amsthm}
\newtheorem{lemma}{Lemma}
\newtheorem{proposition}{Proposition}

\usepackage{cleveref}
\usepackage{setspace}

\begin{document}
\doublespacing

\title{Bivariate Prior Specification for Bayesian Decision-Making in Early-Phase Clinical Trials}

\author[1,*]{Chengyuan Yang}
\author[1]{Evan Kwiatkowski}

\affil[1]{Department of Biostatistics and Data Science,

University of Texas Health Science Center at Houston, 
Houston, TX, USA}

\affil[*]{\texttt{Chengyuan.Yang@uth.tmc.edu}}

\maketitle

\begin{abstract}
{Bayesian Go/No-Go decisions with co-primary 
endpoints require specifying prior distributions under the 
Normal-Inverse-Wishart framework; however guidance on how prior hyperparameters influence trial decisions remains limited. We propose a calibrated prior specification framework for bivariate Go/No-Go decisions. Skeptical and enthusiastic priors are calibrated so that each assigns a target probability to a clinically relevant decision region. We prove that for any prior precision $\kappa > 0$, a unique scale parameter $\lambda_0$ achieves the target calibration. Operating 
characteristics are evaluated across different $\kappa $ via simulation and applied to a phase 3 telitacicept lupus trial. The simulation result indicates $\kappa$ is the primary driver of prior discrimination. At $\kappa = 1$, the go rate difference between priors was 0.07; at $\kappa = 10$ it reached 0.56, with false positive rates below 0.01. Operating characteristics were robust to the degrees of freedom parameter $\nu_0$ and prior correlation $\rho_0$, supporting a default of $\nu_0 = 2$. In the lupus application, prior sensitivity was negligible at $\kappa = 1$ but at $\kappa = 10$ the enthusiastic go rate was three times the skeptical rate at small sample sizes. The framework reduces prior specification to two choices: the prior center and the prior precision $\kappa$. The identification of $\kappa$ as the dominant parameter, together with the 
cautious choice of  $\kappa$ before the trial, motivates adaptive approaches to prior precision.}

\end{abstract}

\newpage
\section{Introduction}
\label{sec1}
Traditional clinical trial designs frequently rely on a single primary endpoint, such as a binary outcome (e.g., Objective Response Rates) or a continuous measure (e.g., survival time). The primary outcome could determine whether a trial proceeds to the next phase by evaluating treatment efficacy and safety. At the end of early phase studies, Go/No-Go (GNG) decisions are commonly informed by proof of concept evaluations, which are specifically designed to assess whether an investigational treatment shows sufficient promise to warrant progression\cite{preskorn2014role}. These evaluations serve as critical decision points to determine whether a candidate should advance to later stage development. However, relying on a single endpoint often overlooks the complexity and heterogeneity of clinical outcomes, and may introduce bias in trial decision
making\cite{mcleod2019choosing}. For instance, endpoints such as time-to-event (TTE), Objective Response Rate (ORR), overall survival (OS), and progression-free survival (PFS), are widely used as the outcome measure in oncology trials to test patient's response to the drug \cite{delgado2021clinical}. While each provides clear and interpretable measures of treatment effect, simply relying on single endpoint may fail to capture the complex nature of disease progression or patient response. In practice, multiple endpoints often coexist and may provide complementary evidence, such as combining PFS and ORR when OS data are not yet mature, or using multiple  outcomes measured at different stages with interim analyses. The growing complexity of clinical trials has highlighted the importance of addressing multiplicity when multiple endpoints are considered\cite{li2017}; this issue has become a critical component of the drug development process, which requires thorough consideration and rigorous evaluation to meet regular standards, such as FDA approval criteria\cite{FDA2023}. 

The use of co-primary endpoints has motivated substantial methodological development. Many trial designs have been proposed to jointly model correlated outcomes, including Bayesian copula-based approaches\cite{nelsen2006introduction,yin2012clinical, cunanan2014evaluating}. Beyond copula-based methods, there are methods that are tailored to specific endpoint types to address co-primary endpoints. Cho et al. proposed a Bayesian design on a cure rate model that incorporated both TTE and PFS endpoints to evaluate cure status, using a bivariate copula with Cholesky decomposition \cite{higham1990analysis} to estimate correlations between endpoints\cite{Cho26012025}. Yang et al. applied a Bayesian model on multinomial distribution to model sequentially occurring co-primary endpoints \cite{Yang2024}. While these advanced methods have strengthened inference for co-primary endpoints, they have mainly focused on model specification and estimation. Zhao et al. proposed BOP2-DC, a Bayesian optimal phase~II design that incorporates both statistical significance 
and clinical relevance into GNG decisions for multiple endpoints\cite{zhao2023bayesian}. However, like earlier approaches, BOP2-DC focuses on optimizing the decision rule structure rather than examining how prior specification on the mean and covariance 
parameters influences the resulting decision. The question of how prior assumptions affect GNG operating characteristics in the co-primary setting has received little attention, despite the fact that in small samples the prior can substantially influence 
the posterior probability on which the decision is based.

Prior assumptions are particularly consequential in early-phase trial designs, where Bayesian methods are attractive for their natural support of probability-based decision rules and their ability to accommodate small sample sizes. In practice, GNG decisions can be based on posterior probabilities exceeding a pre-defined threshold, for instance, the posterior probability of superiority. Under such designs, the prior distribution interacts directly with the decision rule to determine the trial 
outcome, yet existing methods for co-primary endpoints have not systematically examined how prior specification affects these operating characteristics.

A broadly applicable framework for co-primary endpoints is the 
multivariate normal hierarchical model, since many endpoints can 
be represented as latent Gaussian variables. Under this model, 
it is common to place an inverse-Wishart prior on the covariance 
matrix due to conjugacy. However, common prior choices for variance parameters are known 
to impose undesirable constraints. The inverse-gamma prior on 
scalar variances concentrates excessive probability mass near 
zero even under seemingly noninformative specifications~\cite{gelman2006prior}. In practice, it shrinks the variance aggressively, even though it was meant to be vague. The inverse-Wishart prior on covariance matrices couples the variance and correlation components, limiting independent control over each other~\cite{gelman2006prior}. The scaled inverse-Wishart reparameterization~\cite{gelman2013bayesian} offers greater flexibility, but practical guidance on how prior choices affect GNG decisions remains limited, especially in the small-sample settings typical of early-phase trials.

In this work, we address this gap by conducting a systematic evaluation of prior specification for bivariate GNG decisions under the conjugate Normal-Inverse-Wishart framework. We propose a calibration function that links prior hyperparameters to a prespecified level of skepticism or enthusiasm about treatment benefit. 

We then prove existence and uniqueness of the calibrated scale parameter, and identify the prior precision $\kappa$ as the primary driver of prior influence on trial decisions. We apply this framework retrospectively to a phase~3 telitacicept trial for systemic 
lupus erythematosus. The remainder of this paper is organized 
as follows: Section~\ref{sec2} describes the proposed methods, 
Section~\ref{sec3} presents the simulation study, 
Section~\ref{sec4} presents the motivating example, and 
Section~\ref{sec5} concludes with a discussion.

\section{Methods}
\label{sec2}

\subsection{Bivariate Normal-Inverse-Wishart Model}
\label{sec:model}

Assume each subject is randomly assigned to one of two arms, 
$k \in \{A, B\}$, with $n_k$ subjects per arm. Let 
$\mathbf{Y}_{ki} = (Y_{ki1}, Y_{ki2})^\top$ denote the bivariate 
outcome for subject $i$ in arm $k$, where
\begin{align}
\mathbf{Y}_{ki} \sim \mathcal{N}(\bm{\mu}_k, \Sigma), \quad 
\bm{\mu}_k = \begin{pmatrix} \mu_{k1} \\ \mu_{k2} \end{pmatrix}, \quad
\Sigma = \begin{pmatrix}
\sigma_1^2 & \rho\sigma_1\sigma_2 \\
\rho\sigma_1\sigma_2 & \sigma_2^2
\end{pmatrix},
\end{align}
and observations are independent across arms. We adopt a conjugate 
Normal-Inverse-Wishart (NIW) prior for the arm-specific means and the 
shared covariance matrix:
\begin{align}
\bm{\mu}_k \mid \Sigma &\sim 
  \mathcal{N}\!\left(\bm{\mu}_{k0},\; \tfrac{1}{\kappa}\,\Sigma\right), 
  \qquad
\Sigma \sim \mathrm{InvWishart}(\nu_0,\;  \Lambda),
\label{eq:prior}
\end{align}
where $\bm{\mu}_{k0}$ is the prior mean for arm $k$, and $\kappa > 0$ 
is the prior precision, shared across both arms 
($\kappa_A = \kappa_B \equiv \kappa$). Higher value of $\kappa$ indicates stronger prior confidence\cite{powerprior_rpackage}. The prior scale matrix is 
decomposed as $\Lambda = \lambda_0 \cdot \Lambda_0$, where
\begin{align}
\Lambda_0 = \begin{pmatrix} 1 & \rho_0 \\ \rho_0 & 1 \end{pmatrix}
\end{align}
encodes the prior correlation structure, and $\lambda_0 > 0$ is a scalar 
that controls the magnitude of the prior variance. This decomposition 
allows $\lambda_0$ to be tuned independently of the correlation structure.

Under this prior, the marginal distribution of the treatment effect 
$\bm{\delta} = \bm{\mu}_A - \bm{\mu}_B$ follows a bivariate 
$t$-distribution:
\begin{align}
\bm{\delta} \sim t_{\nu_0+1-d}\!\left(
  \bm{\delta}_0,\; 
  \tfrac{2\lambda_0}{\kappa(\nu_0+1-d)}\,\Lambda_0 
\right),
\label{eq:delta distribution}
\end{align}
where $\bm{\delta}_0 = \bm{\mu}_{A0} - \bm{\mu}_{B0}$ and $d = 2$. 
The factor $2/\kappa$ in the scale reflects the equal-precision 
assumption across arms; the extension to unequal precisions replaces 
$\kappa/2$ with the harmonic mean 
$(\kappa_A^{-1} + \kappa_B^{-1})^{-1}$.

After observing data, the posterior distributions follow standard 
conjugate updating rules:
\begin{align}
\bm{\mu}_k \mid \Sigma, \mathbf{Y} &\sim 
  \mathcal{N}\!\left(
    \tfrac{\kappa\,\bm{\mu}_{k0} + n_k\,\bar{\mathbf{Y}}_k}
          {\kappa + n_k},\; 
    \tfrac{1}{\kappa + n_k}\,\Sigma
  \right), 
\label{eq:post mean} \\[4pt]
\Sigma \mid \mathbf{Y}, \bm{\mu} &\sim 
  \mathrm{InvWishart}\!\left(
    \nu_0 + n_A + n_B,\; 
    \lambda_0\Lambda_0 + S_A + S_B
  \right),
\label{eq:post variance}
\end{align}
where $\bar{\mathbf{Y}}_k = n_k^{-1}\sum_{i=1}^{n_k}\mathbf{Y}_{ki}$ 
is the sample mean for arm $k$, and $S_k$ is the within-arm scatter 
matrix (definitions in Appendix~\ref{sup: proof}). The posterior mean 
for each arm is a weighted average of the prior mean and the sample 
mean, with weight $\kappa/(\kappa + n_k)$ on the prior. Derivations 
are given in Appendix~\ref{sup: proof}.

\subsection{Decision rules}
\label{sec:decisionrules}
Based on the posterior distributions derived in Section~\ref{sec:model}, 
we consider a superiority decision rule comparing the mean efficacy between treatment arms \(A\) and \(B\).
The hypotheses are formulated as
\begin{align}
H_0: \bm{\mu}_A \le \bm{\mu}_B 
\quad \text{versus} \quad
H_1: \bm{\mu}_A > \bm{\mu}_B,
\end{align}
where the inequalities are interpreted componentwise.

Let \(\mathbf{Y}_A\) and \(\mathbf{Y}_B\) denote the observed bivariate outcomes for the two arms.  
The posterior probability of superiority is then 
\begin{equation}
\begin{aligned}
P(\bm{\mu}_A > \bm{\mu}_B \mid \mathbf{Y}_A, \mathbf{Y}_B)
&= P(\mu_{A1} > \mu_{B1},\, \mu_{A2} > \mu_{B2} \mid \mathbf{Y}_A, \mathbf{Y}_B) \\
&= \int_{-\infty}^{\infty} \!\!\int_{-\infty}^{\infty} 
    \!\!\int_{\mu_{B2}}^{\infty} \!\!\int_{\mu_{B1}}^{\infty}
    f(\bm{\mu}_A, \bm{\mu}_B \mid \mathbf{Y}_A, \mathbf{Y}_B)
    \, d\mu_{A1}\, d\mu_{A2}\, d\mu_{B1}\, d\mu_{B2}.
\end{aligned}
\end{equation}
In practice, this is equivalent to $P(\bm{\delta} > \boldsymbol{\tau} \mid \text{data})$ where $\bm{\delta}$ follows the posterior marginal, and $\tau$ is the decision threshold vector.

Then, to induce a dichotomous decision on arm $A$ is more efficacious than arm $B$, we use an indicator function with a threshold $c$:
\begin{equation}
\varphi(\mathbf{y}) = 
\begin{cases}
1 & \text{if } P(\bm{\delta} > \boldsymbol{\tau} \mid \text{data}) > c \\
0 & \text{if } P(\bm{\delta} > \boldsymbol{\tau} \mid \text{data}) \le c,
\end{cases}
\label{decision equation}
\end{equation}
where 1 and 0 indicate the rejection and acceptance of the null hypothesis, respectively.

\subsection{Enthusiastic and Skeptical Priors}
\label{sec:prior_region}

We characterize prior informativeness through two contrasting specifications: the \emph{skeptical prior} and the \emph{enthusiastic prior}. Each is defined by its center $\bm{\delta}_0$ in the treatment-effect space and by the prior 
probability it assigns to a clinically relevant region before observing any data.

The skeptical prior is centered at or near the decision threshold $\boldsymbol{\tau}$, reflecting the belief that the treatment effect is unlikely to exceed the threshold for clinical relevance. Under this specification, the most prior mass concentrates near the threshold, and relatively little mass falls in the decision region $\{\bm{\delta} : \bm{\delta} > \boldsymbol{\tau}\}$ (Figure~\ref{fig:prior_regions}, right panel)

Conversely, the enthusiastic prior is centered at a positive effect above the threshold, reflecting optimism about treatment benefit. Under this specification, relatively little prior mass falls in the 
region of harm $\{\bm{\delta}:\bm{\delta} < \boldsymbol{\tau}\}$ 
(Figure~\ref{fig:prior_regions}, left panel).

To calibrate the scale parameter $\lambda_0$ for each prior, we 
define the calibration function
\begin{align}
\label{eq:calibration function}
f(\lambda_0;\, \nu_0, \kappa, \bm{\delta}_0, \boldsymbol{\tau})
  = P\!\bigl(\bm{\delta} \in \mathcal{R}^*\bigr),
\end{align}
where $\mathcal{R}^*$ denotes the relevant decision region and 
the probability is evaluated under the prior distribution Equation \eqref{eq:delta distribution}. In both cases, we solve for $\lambda_0$ such that $f(\lambda_0) = \alpha$, where $\alpha$ is a small target 
probability (e.g., 0.05). This ensures that both priors assign 
the same controlled amount of mass to their respective tail 
regions, providing a fair basis for comparison.

\begin{figure}
    \centering
    \includegraphics[width=0.8\linewidth]{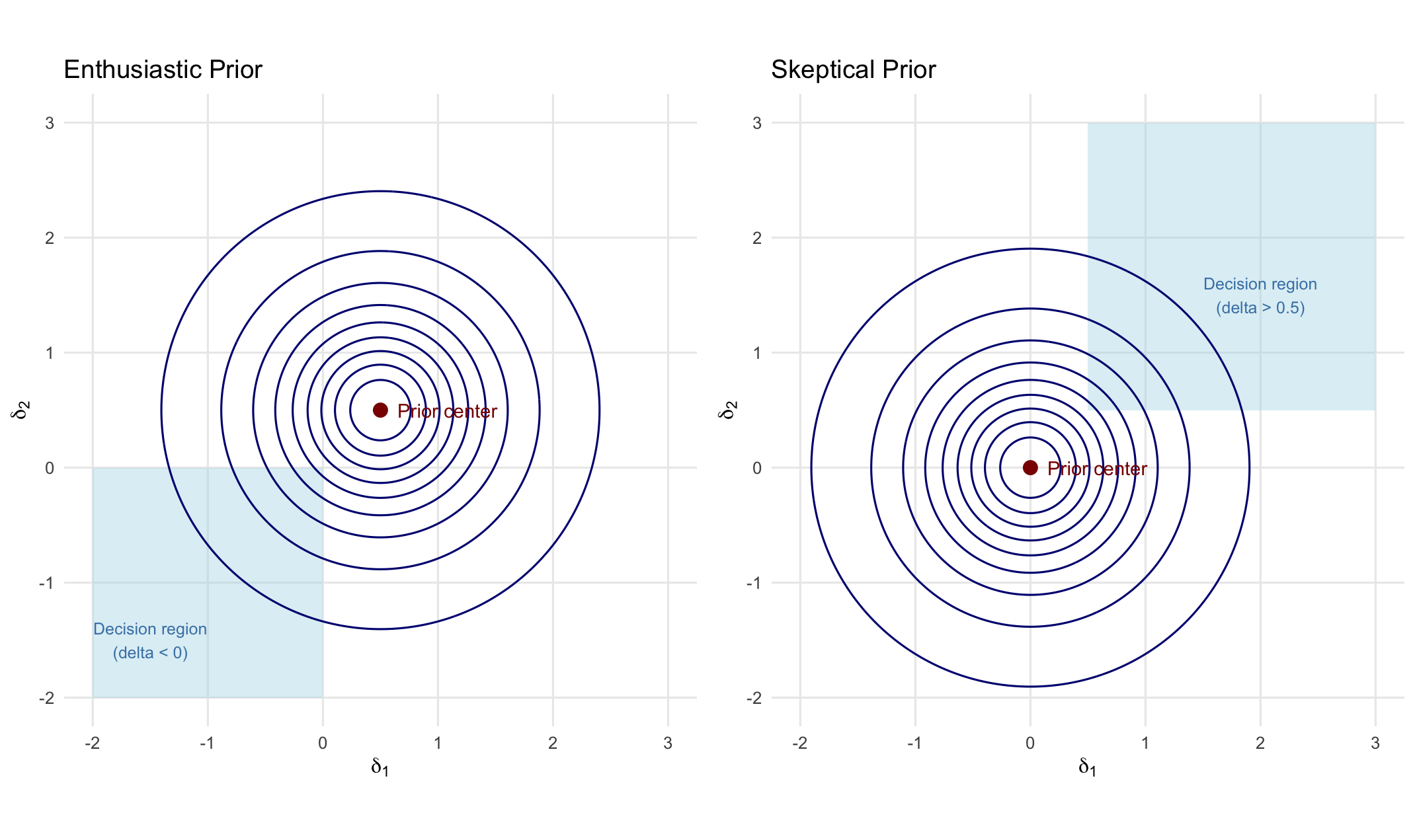}
    \caption{Illustration of the enthusiastic prior (left) centered at 
$\bm{\delta}_0 = (0.5, 0.5)$ and the skeptical prior (right) centered 
at $\bm{\delta}_0 = (0, 0)$, with $\nu_0 = 2$ and $\rho_0 = 0$. Contours show the bivariate $t$ density of the marginal prior on $\bm{\delta}$. Shaded regions indicate the decision regions used for calibration. The scale parameter $\lambda_0$ is enlarged for visual clarity; calibrated values are reported in Supplementary
Table~\ref{tab:s1}.}
    \label{fig:prior_regions}
\end{figure}
\FloatBarrier

\subsection{Recommendation of default hyperparameters}
\label{sec:default hyper}
The marginal prior on $\bm{\delta}$ in 
Equation~\eqref{eq:delta distribution} depends on three 
hyperparameters: the scale $\lambda_0$, the degrees of freedom 
$\nu_0$, and the prior precision $\kappa$. To provide a principled 
default specification, we fix $\nu_0$ and establish that for any 
$\kappa > 0$, a unique $\lambda_0$ can be calibrated to achieve a 
target level of prior skepticism or enthusiasm. 

\begin{proposition}[]
\label{prop:existence} 
Let 
$f(\lambda_0;\, \nu_0, \kappa, \bm{\delta}_0, \boldsymbol{\tau})$ 
denote the calibration function~\eqref{eq:calibration function}, 
defined as the prior probability that the treatment effect 
$\bm{\delta}$ falls in the decision region when the prior is 
centered at $\bm{\delta}_0$ with threshold $\boldsymbol{\tau}$. 
For the skeptical prior, 
$f = P(\bm{\delta} > \boldsymbol{\tau} \mid \bm{\delta}_0)$; 
for the enthusiastic prior, 
$f = P(\bm{\delta} < \boldsymbol{\tau} \mid \bm{\delta}_0)$.

For fixed $\nu_0 > d - 1$, fixed $\kappa > 0$, and 
$\bm{\delta}_0 \neq \boldsymbol{\tau}$, 
$f(\lambda_0)$ is continuous and strictly monotone in $\lambda_0$ 
on $(0, \infty)$ with boundary limits
\[
\lim_{\lambda_0 \to 0} f(\lambda_0) \in \{0, 1\}, 
\qquad 
\lim_{\lambda_0 \to \infty} f(\lambda_0) = L,
\]
where $L$ is the orthant probability of the centered bivariate $t$ 
distribution with correlation $\rho_0$ and $\nu_0 - d + 1$ degrees 
of freedom. 

Consequently, for any target probability 
$\alpha \in (0, L)$, there exists a unique 
$\lambda_0^* > 0$ satisfying $f(\lambda_0^*) = \alpha$. This 
guarantee holds for every $\kappa > 0$: changing $\kappa$ alters 
the scale of the prior on $\bm{\delta}$ through 
Equation~\eqref{eq:delta distribution}, but the calibration function 
remains continuous and monotone in $\lambda_0$, so a unique 
recalibrated $\lambda_0^*(\kappa)$ exists at each $\kappa$.

The proof is provided in Appendix~\ref{app: proposition proof}.
\end{proposition}

This proposition ensures that the calibration framework is valid across the entire range of $\kappa$. For $\nu_0$, we recommend the default $\nu_0 = \nu_{\min}$, the 
smallest integer satisfying $\nu_0 - d + 1 > 0$, which yields the 
heaviest-tailed and therefore least informative prior within the calibrated family. In the bivariate case ($d = 2$), this gives $\nu_0 = 2$. The simulation results in Supplementary table ~\ref{tab:s2} confirm that operating characteristics are robust to $\nu_0$, with 
correct prior ordering preserved across 
$\nu_0 \in \{2, 4, 10, 50\}$, supporting this default. 

With $\nu_0$ fixed, prior specification reduces to two interpretable choices: the prior center $\bm{\delta}_0$, which encodes the direction and 
magnitude of prior belief about treatment effect, and $\kappa$, controlling how strongly that belief influences the posterior.

\section{Simulations}
\label{sec3}
\subsection{Operating characteristics}
\label{sec:oc_sim}

We simulated a two-arm trial with $n = 10$ subjects per arm. Two 
correlated endpoints were generated from a bivariate normal 
distribution with standard deviation~1 for both endpoints. The 
true correlation was set at $\rho_{\mathrm{true}} \in \{0, 0.6\}$, 
and the prior correlation at $\rho_0 \in \{0, 0.6\}$. True treatment effects were varied across four scenarios (Table~\ref{tab:scenarios}); we report the null and both-effects scenarios in the main text.

For the skeptical prior, both arm-level prior means were set to be equal, giving $\bm{\delta}_0 = (0,0)$. For the enthusiastic prior, the treatment arm prior mean was shifted so that $\bm{\delta}_0 = (0.5, 0.5)$. In both cases, $\lambda_0$ was 
calibrated with a fixed $\alpha$, e.g. 0.05,  in the target tail region. The degrees of freedom were fixed at the default $\nu_0 = 2$ throughout the main analysis. The prior precision was varied across 
$\kappa \in \{1, 5, 10\}$, with $\lambda_0$ recalibrated at each $\kappa$ via Proposition~\ref{prop:existence} to maintain the target tail probability. For each configuration, $10{,}000$ trials were simulated. A Go 
decision was declared if the posterior probability 
$P(\bm{\delta} > \boldsymbol{\tau} \mid \text{data})$ exceeded $c = 0.95$. The Go rate under the null scenario estimates the false positive rate, while under alternative scenarios it estimates power.

\begin{table}[htbp]
\centering
\begin{tabular}{clp{6cm}}
\hline
Scenario & Description of treatment arm & Mean structure (treatment $A$ vs control $B$) \\
\hline
1 & Null scenario
  &$\mu_{A1} - \mu_{B1} = \mu_{A2} - \mu_{B2} = 0$ \\

2 & Both endpoints have positive effects 
  &$\mu_{A1} - \mu_{B1} = \mu_{A2} - \mu_{B2}>0$ \\

3 & One endpoint effective, one nearly null 
  & $\mu_{A1} - \mu_{B1} > 0,\quad \mu_{A2} - \mu_{B2} \approx 0$ 
    \\

4 & Mixed magnitude positive effects 
  & $\mu_{A1} - \mu_{B1} = 2(\mu_{A2} - \mu_{B2})$ \\

\hline
\end{tabular}
\caption{Data-generating scenarios for co-primary endpoints under two arms. Across all scenarios, treatment means $\mu_A\in \{0,0.5,1\}$, while control means $\mu_B$ are fixed as 0 to define the no effect.}
\label{tab:scenarios}
\end{table}

\subsection{Results}
Table \ref{tab:priors} reports the calibrated $
\lambda_0$ values for each $\kappa$. The calibrated values confirm Proposition \ref{prop:existence} that $\lambda_0$ increases with $\kappa$ proportionally to compensate for the tighter prior scale. Figure~\ref{fig:prior_skep_heatmap} visualizes the calibration function for the skeptical prior across $(\kappa, \lambda_0)$ with $\nu_0 = 2$ fixed. The contour lines trace the unique $(\kappa, \lambda_0)$ pairs satisfying $f(\lambda_0) = \alpha$ for $\alpha = 0.05$ and $\alpha = 0.01$, illustrating the monotone relationship established in 
Proposition~\ref{prop:existence}. Figure~\ref{fig:prior_skep_heatmap} shows that as $\lambda_0 \to \infty$, the calibration function converges to the orthant probability $L = 1/4 + \arcsin(\rho_0)/(2\pi)$, which equals $0.25$ when $\rho_0 = 0$ and increases to approximately $0.35$ when $\rho_0 = 0.6$(Lemma~\ref{lem:limits} in Appendix~\ref{app: proposition proof}). By the 
symmetry of the calibration function, the enthusiastic prior 
produces identical $\lambda_0$ values at each $\kappa$, 
as shown in Figure~\ref{fig:prior_enth_heatmap}.

\begin{table}[ht]
\centering
\caption{Prior specifications for the skeptical and enthusiastic 
         priors. For each $\kappa$, $\lambda_0$ is calibrated via 
         the calibration function~\eqref{eq:calibration function} 
         to achieve a target tail probability of $\alpha = 0.05$. 
         All specifications use $\nu_0 = 2$ and $\rho_0 = 0$.}
\label{tab:priors}
\begin{tabular}{lcccccc}
\toprule
& & & \multicolumn{3}{c}{Calibrated $\lambda_0$} \\
\cmidrule(lr){4-6}
Prior & Center $\bm{\delta}_0$ & Threshold $\boldsymbol{\tau}$ 
      & $\kappa = 1$ & $\kappa = 5$ & $\kappa = 10$ \\
\midrule
Skeptical    & $(0,\, 0)$     & $(0.5,\, 0.5)$ 
             & 0.019 & 0.097  &0.194 \\
Enthusiastic & $(0.5,\, 0.5)$ & $(0,\, 0)$     
             & 0.019 & 0.097  &0.194 \\
\bottomrule
\end{tabular}
\begin{tablenotes}
\small
\item By the symmetry of the calibration function, both priors 
      yield identical $\lambda_0$ at each $\kappa$. Target tail 
      probability: $P(\bm{\delta} > \boldsymbol{\tau}) = 0.05$ 
      for the skeptical prior; 
      $P(\bm{\delta} < \boldsymbol{\tau}) = 0.05$ for the 
      enthusiastic prior.
\end{tablenotes}
\end{table}

\begin{figure}[ht]
    \centering
    \begin{subfigure}[b]{0.8\textwidth}
        \centering
        \includegraphics[width=\textwidth]{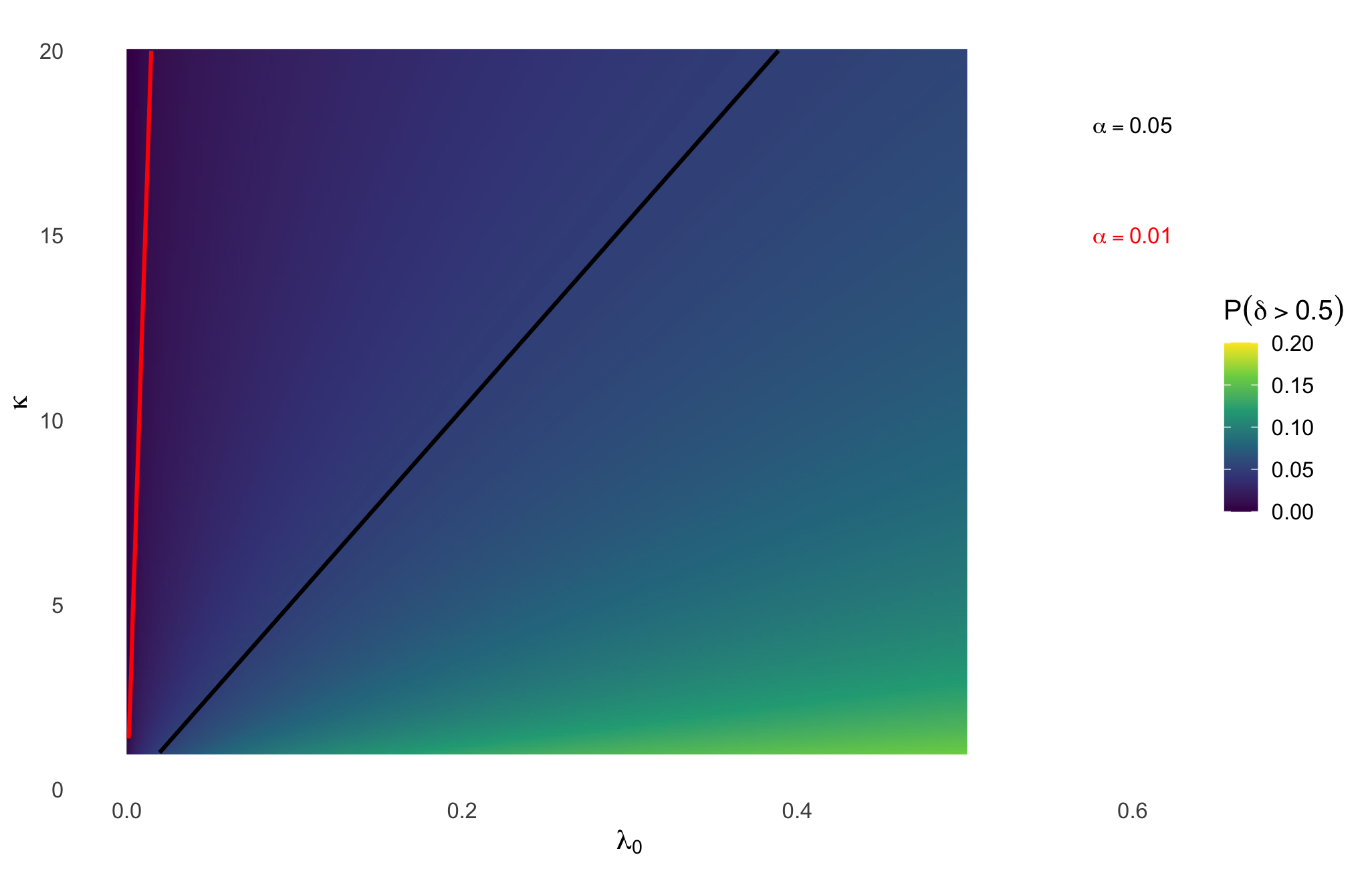}
        \caption{$\rho_0 = 0$}
        \label{fig:heatmap_rho0}
    \end{subfigure}
    \hfill
    \begin{subfigure}[b]{0.8\textwidth}
        \centering
        \includegraphics[width=\textwidth]{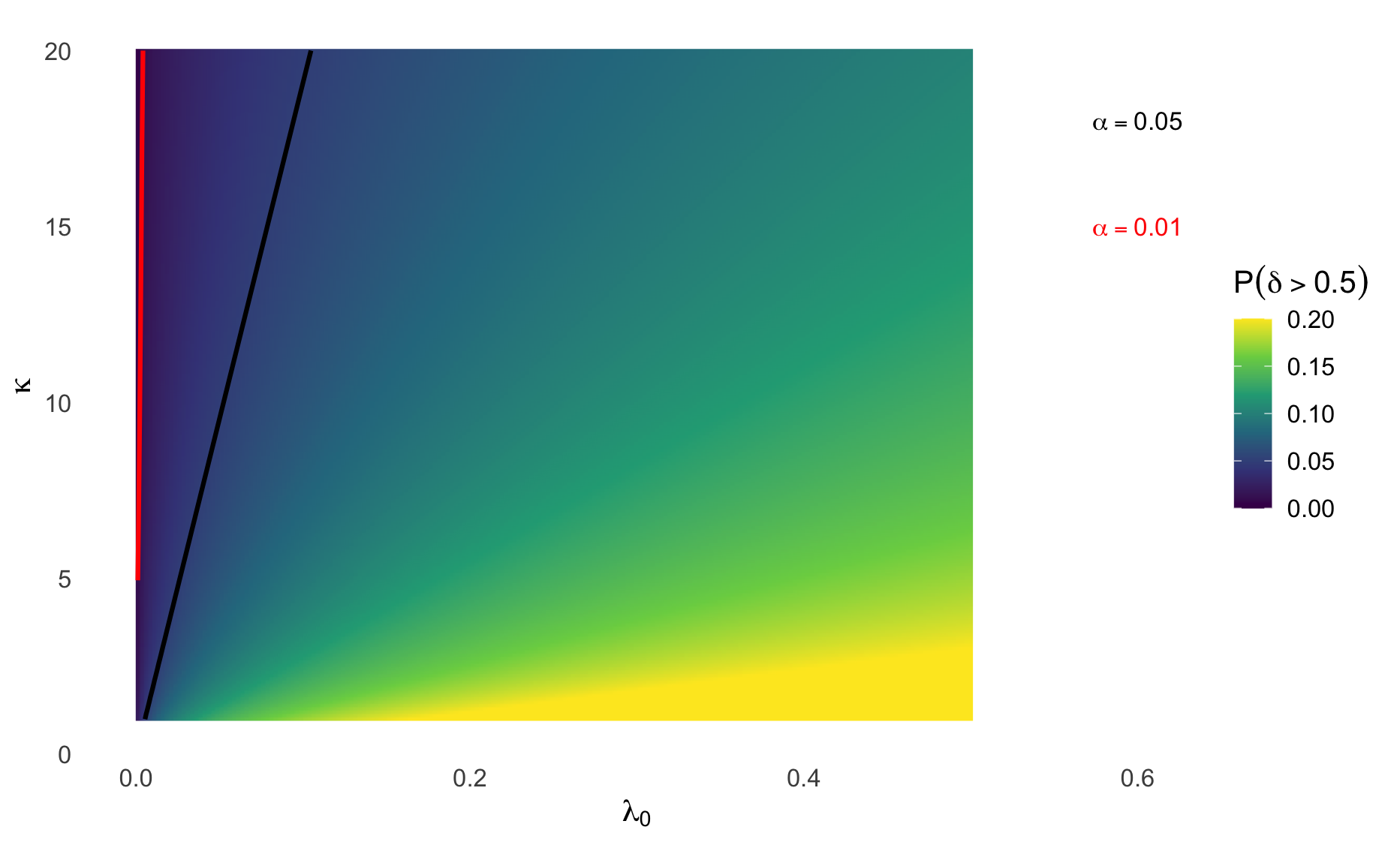}
        \caption{$\rho_0 = 0.6$}
        \label{fig:heatmap_rho06}
    \end{subfigure}
    \caption{Calibration function $P(\bm{\delta} > \boldsymbol{\tau})$ 
             for the skeptical prior across $(\kappa, \lambda_0)$ with 
             $\nu_0 = 2$. Contour lines indicate $\alpha = 0.05$ 
             (black) and $\alpha = 0.01$ (red). The contour lines 
             trace the $(\kappa, \lambda_0)$ pairs satisfying the 
             calibration in Proposition~\ref{prop:existence}.}
    \label{fig:prior_skep_heatmap}
\end{figure}

\FloatBarrier

Table~\ref{tab:kappa_sensitivity}  reports go rates for skeptical and enthusiastic priors across $\kappa \in \{1,5, 10\}$, with $\nu_0$ fixed at its default value of 2 and  recalibrated $\lambda_0$ at each $\kappa$ to maintain the 0.05 target tail probability. The results demonstrate that $\kappa$ is the primary driver of prior distinction in this framework. Especially looking at the strong effect size case, at $\kappa = 1$, the difference between enthusiastic and skeptical go rates is only about 0.07. Increasing to $\kappa = 5$ raises this difference to 0.324, and at $\kappa = 10$ 
it reaches 0.57, where a skeptical prior yields a go rate of only 0.07 compared to 0.63 for the enthusiastic prior. Thus, although this recalibration maintains the same target tail probability across all $\kappa$ values, it does not offset the direct effect of $\kappa$ on the posterior mean weight $\kappa/(\kappa + n)$. 

This amplification is explained by the posterior mean weight. A larger $\kappa$ pulls the posterior mean more strongly toward the prior center, producing 
systematically higher Go rates for the enthusiastic prior and lower Go rates for the skeptical prior from the same observed data. Importantly, false positive rates remain controlled across all 
$\kappa$ values, ( less than 0.01 for the null scenario), confirming that increased prior discrimination does not come at the cost of inflated false positive rate.

Increasing sample size raises go rates substantially and reduces prior sensitivity. At $n = 50$ under the strong effect, go rates exceed 0.99 for both priors regardless of $\kappa$, indicating that the data dominate the posterior at this sample 
size. The $\kappa$ effect is most pronounced at $n = 10$, where prior beliefs carry the greatest relative weight.

True endpoint correlation also increases go rates. Under $\rho_{\mathrm{true}} = 0.6$, go rates are approximately 5 to 10 percentage points higher than under independence across all settings, reflecting the additional information gained from correlated endpoints. However, the prior correlation assumption $\rho_0$ has minimal impact. Results under $\rho_0 = 0.6$ (Supplementary Table~\ref{tab:kappa_sensitivity_rho06}) are similar to those under $\rho_0 = 0$, indicating that misspecification of the prior correlation does not affect operating characteristics. Importantly, false positive rates remain below 0.01 across all $\kappa$ values and correlation structures, confirming that increased prior 
discrimination does not inflate type~I error.

\begin{table}[ht]
\centering
\caption{Go rates for skeptical and enthusiastic priors across 
         $\kappa \in \{1, 5, 10\}$ and sample sizes, under the 
         both-effects scenario, $\nu_0 = 2$, $c = 0.95$, 
         $10{,}000$ replications. For each $\kappa$, $\lambda_0$ 
         is recalibrated to maintain the target tail probability 
         of $0.05$.}
\label{tab:kappa_sensitivity}
\begin{tabular}{llcccccc}
\toprule
& & \multicolumn{3}{c}{$\rho_0 = 0,\; \rho_{\mathrm{true}} = 0$} 
  & \multicolumn{3}{c}{$\rho_0 = 0,\; \rho_{\mathrm{true}} = 0.6$} \\
\cmidrule(lr){3-5} \cmidrule(lr){6-8}
$\kappa$ & Prior & Null & Moderate & Strong 
                  & Null & Moderate & Strong \\
\midrule
\multicolumn{8}{l}{\textit{$n = 10$ per arm}} \\
\addlinespace
\multirow{2}{*}{1}
  & Skeptical    & 0.0021 & 0.0658 & 0.4417 & 0.0122 & 0.1429 & 0.5269 \\
  & Enthusiastic & 0.0027 & 0.0902 & 0.5097 & 0.0168 & 0.1750 & 0.5889 \\
\addlinespace
\multirow{2}{*}{5}
  & Skeptical    & 0.0004 & 0.0228 & 0.2263 & 0.0043 & 0.0659 & 0.3416 \\
  & Enthusiastic & 0.0048 & 0.1158 & 0.5505 & 0.0229 & 0.2102 & 0.6245 \\
\addlinespace
\multirow{2}{*}{10}
  & Skeptical    & 0.0000 & 0.0035 & 0.0696 & 0.0013 & 0.0204 & 0.1598 \\
  & Enthusiastic & 0.0085 & 0.1703 & 0.6321 & 0.0378 & 0.2857 & 0.6887 \\
\midrule
\multicolumn{8}{l}{\textit{$n = 25$ per arm}} \\
\addlinespace
\multirow{2}{*}{1}
  & Skeptical    & 0.0017 & 0.2547 & 0.9265 & 0.0101 & 0.3562 & 0.9387 \\
  & Enthusiastic & 0.0025 & 0.2869 & 0.9387 & 0.0125 & 0.3881 & 0.9505 \\
\addlinespace
\multirow{2}{*}{5}
  & Skeptical    & 0.0005 & 0.1925 & 0.8923 & 0.0069 & 0.2957 & 0.9144 \\
  & Enthusiastic & 0.0039 & 0.3467 & 0.9553 & 0.0188 & 0.4474 & 0.9638 \\
\addlinespace
\multirow{2}{*}{10}
  & Skeptical    & 0.0002 & 0.1320 & 0.8445 & 0.0035 & 0.2291 & 0.8683 \\
  & Enthusiastic & 0.0070 & 0.4405 & 0.9717 & 0.0317 & 0.5318 & 0.9745 \\
\midrule
\multicolumn{8}{l}{\textit{$n = 50$ per arm}} \\
\addlinespace
\multirow{2}{*}{1}
  & Skeptical    & 0.0014 & 0.6114 & 0.9992 & 0.0099 & 0.6702 & 0.9994 \\
  & Enthusiastic & 0.0016 & 0.6351 & 0.9995 & 0.0113 & 0.6914 & 0.9995 \\
\addlinespace
\multirow{2}{*}{5}
  & Skeptical    & 0.0012 & 0.5742 & 0.9986 & 0.0082 & 0.6364 & 0.9990 \\
  & Enthusiastic & 0.0031 & 0.6954 & 0.9997 & 0.0164 & 0.7356 & 0.9995 \\
\addlinespace
\multirow{2}{*}{10}
  & Skeptical    & 0.0005 & 0.5313 & 0.9980 & 0.0064 & 0.5986 & 0.9980 \\
  & Enthusiastic & 0.0054 & 0.7607 & 0.9997 & 0.0260 & 0.7902 & 0.9998 \\
\bottomrule
\end{tabular}
\begin{tablenotes}
\small
\item $\rho_0$: prior correlation in $\Lambda_0$; 
      $\rho_{\mathrm{true}}$: true correlation used to generate 
      data. Null: $\bm{\delta} = (0,0)$; Moderate : 
      $\bm{\delta} = (0.5, 0.5)$; Strong: 
      $\bm{\delta} = (1, 1)$.
\end{tablenotes}
\end{table}
\FloatBarrier

\section{Motivating example: the Telitacicept Lupus Trial}
\label{sec4}
To illustrate the proposed Bayesian decision framework in a real clinical setting, we applied it retrospectively to a phase 3 trial of telitacicept for systemic lupus erythematosus \cite{van2025phase}. The trial randomized 167 participants to the telitacicept arm and 168 to placebo. The primary endpoint was a binary SRI-4 composite responder rate at week 52, defined as a simultaneous reduction of at least 4 points from baseline in the SELENA-SLEDAI score, no new disease activity as measured by the BILAG index\cite{isenberg2005bilag, yee2009bilag}, and no worsening in the Physician’s Global Assessment (PGA) score\cite{petri2005combined}. The secondary endpoints included the continuous component measures of the SRI-4 assessed separately: the mean change from baseline in SELENA-SLEDAI score and the mean change from baseline in PGA score, both at week 52. For this application, we focus on two continuous secondary endpoints. Since the original trial reports treatment improvement as a negative change, we flipped the sign of all outcomes so that a positive change is consistent with the decision rule. 

\subsection*{Data generation} The true treatment means on the flipped scale were taken directly from the paper: $\mu_T$ = (4.95, 0.79) for the telitacicept arm and $\mu_C$ = (1.00, 0.40) for the placebo arm, corresponding to SELENA-SLEDAI and PGA respectively. Since individual patient-level data were unavailable, the endpoint standard deviations were recovered from the reported Wald 95 confidence intervals using the estimation, yielding $\sigma_T$ = (19.3, 0.69) and $\sigma_C$ = (22.0, 0.51) for the treatment and placebo arms, respectively.

Post-treatment SELENA-SLEDAI scores were generated from a lognormal 
distribution to reflect the non-negative and right-skewed nature of 
the score. Specifically, baseline scores $X_0 \sim 
\text{LogNormal}(\mu_{0,\log}, \sigma_{0,\log}^2)$ and post-treatment 
scores $X_1 \sim \text{LogNormal}(\mu_{1,\log}, \sigma_{1,\log}^2)$ 
were generated independently, with the change score computed as 
$\Delta = -(X_1 - X_0)$ on the flipped scale. Lognormal parameters 
were derived from:
\begin{align*}
\mu_{\log} = \log(E[X]) - \frac{\sigma_{\log}^2}{2}, 
\quad \sigma_{\log}^2 &= \log\left(1 + \frac{\text{SD}^2}{E[X]^2}\right), 
\end{align*}
yielding $(\mu_{0,\log}, \sigma_{0,\log}^2) = (2.407, 0.070)$ for 
baseline and $(\mu_{1T,\log}, \sigma_{1T,\log}^2) = (0.755, 2.247)$ 
for the treatment arm post-treatment score. This approach improved 
the simulated SELENA-SLEDAI responder rate from 52\% under normality 
to 77\%, closer to the observed 70.1\%. PGA changes were generated 
from a normal distribution. The analysis model assumed bivariate 
normality throughout, consistent with the NIW conjugate framework.

\subsection*{Correlation Estimation}
Since this trial does not directly report the correlation between SELENA-SLEDAI and PGA score, we estimated the correlation based on the binary response and the marginal responder rate for each outcome individually. Assuming the BILAG component of SRI-4 is independent to SELENA-SLEDAI and PGA, the SRI-4 responder probability can be approximated as the joint probability $P(X_{\text{SELENA}} \geq 4 \ \text{ and } \ X_{\text{PGA}} \geq -0.3) = \Phi_2(-c_1, -c_2;\, \rho)$, where $\Phi_2$ is the standard bivariate normal CDF with correlation $\rho$. The threshold $c_1 \approx -0.527$ was derived from the observed marginal SELENA-SLEDAI rate of 0.701, and $c_2\approx -1.574$ was derivedd from the mean PGA change. Setting $\Phi_2(-c_1, -c_2;\, \rho) = 0.671 $ and solving numerically via bisection, we obtain $\hat{\rho} = 0.239$.

\subsection*{Results}
Table \ref{tab:lupus_kappa} reports go rates for the telitacicept lupus trial application across prior precision values $\kappa \in \{1, 5, 10\}$ and sample sizes $n \in \{20, 50, 80, 167\}$ per arm, with superiority threshold. Prior centers were matched to the SRI-4 clinical efficacy criteria. The skeptical prior was centered at $(4,0.3)$, corresponding to the minimum clinically meaningful threshold, and the enthusiastic prior at $(8,0.6)$, representing a treatment effect twice the minimum threshold. 

At the default prior precison $\kappa = 1$ per arm, go rates are very similar between enthusiastic prior and skeptical prior, at most 0.013 at $n = 20$ and drop to 0.005 at full sample size $n = 167$. This mirrors the general simulation finding that $\kappa = 1$ produces negligible prior sensitivity. 
At $\kappa = 5$, the enthusiastic prior yields a go rate of 0.165 at $n = 20$
compared to 0.087 for the skeptical prior, a difference of 0.078. At
$\kappa = 10$ and $n = 20$, this gap widens to 0.189 with the enthusiastic go
rate (0.282) roughly three times to the skeptical go rate (0.093). This
shows that meaningful prior discrimination emerges in the lupus
application only when $\kappa$ is sufficiently large relative to the sample
size, confirming $\kappa$ as the primary driver of prior influence identified
in the general simulation.

The interaction between $\kappa$ and sample size is also apparent. At $n = 167$, even $\kappa = 10$ produces a modest difference of 0.038 between priors (0.652 versus 0.614), because the large sample size dilutes the prior's contribution to the posterior.

\begin{table}[ht]
\centering
\caption{Go rates for skeptical and enthusiastic priors under the 
         telitacicept lupus trial application across prior precision 
         values $\kappa \in \{1, 5, 10\}$ and 
         sample sizes, $\nu_0 = 2$, $\rho_0 = 0$, 
         $\rho_{\text{true}} = 0$, using control 
         $\boldsymbol{\mu}_C = (1.00, 0.40)$, superiority threshold 
         $\boldsymbol{\delta} > (0, 0)$, $c = 0.95$, 
         10{,}000 replications.}
\label{tab:lupus_kappa}
\begin{tabular}{llcccc}
\toprule
& & \multicolumn{4}{c}{Go Rate} \\
\cmidrule(lr){3-6}
$\kappa$ & Prior & $n = 20$ & $n = 50$ & $n = 80$ & $n = 167$ \\
\midrule
\multirow{2}{*}{1}
  & Skeptical    & 0.0795 & 0.2492 & 0.3764 & 0.5843 \\
  & Enthusiastic & 0.0924 & 0.2618 & 0.3849 & 0.5892 \\
\addlinespace
\multirow{2}{*}{5}
  & Skeptical    & 0.0874 & 0.2611 & 0.3947 & 0.5960 \\
  & Enthusiastic & 0.1652 & 0.3224 & 0.4382 & 0.6170 \\
\addlinespace
\multirow{2}{*}{10}
  & Skeptical    & 0.0932 & 0.2853 & 0.4097 & 0.6135 \\
  & Enthusiastic & 0.2820 & 0.4062 & 0.4927 & 0.6518 \\
\bottomrule
\end{tabular}
\begin{tablenotes}
\small
\item Skeptical prior centred at $(4,\ 0.3)$ with target probability 
      0.05 in the upper tail; enthusiastic prior centred at 
      $(8,\ 0.6)$ with target probability 0.05 in the lower tail, 
      corresponding to the SRI-4 minimum efficacy threshold and 
      twice the threshold respectively. Data generated from a 
      lognormal distribution for SELENA-SLEDAI and normal 
      distribution for PGA. True treatment mean 
      $\boldsymbol{\mu}_T = (4.95,\ 0.79)$.
\end{tablenotes}
\end{table}

\FloatBarrier

\section{Discussion}
\label{sec5}
We have proposed a calibrated prior specification framework for bivariate 
Bayesian Go/No-Go decisions under the Normal--Inverse--Wishart model, and 
systematically examined how the hyperparameters $\lambda_0$, $\nu_0$, and $\kappa$ influence prior informativeness and posterior operating 
characteristics. The calibration function links prior hyperparameters to a prespecified clinical statement (e.g. a specified level of prior skepticism or enthusiasm about treatment benefit) through a linked pair $(\kappa,\lambda_0)$. The prior precision 
$\kappa$ controls how strongly prior beliefs influence the 
posterior through the weight on the prior 
mean, while $\lambda_0$ is recalibrated at each $\kappa$ via Proposition~\ref{prop:existence} to preserve the target tail probability. The degrees of freedom parameter $\nu_0$ governs the tail shape of the prior but, as confirmed by simulation, has minimal impact on operating characteristics, supporting choosing the default value of $\nu_0$.

In contrast, the prior precision $\kappa$ emerged as the primary 
driver of prior discrimination. Although $\lambda_0$ is 
recalibrated at each $\kappa$ to maintain the target tail 
probability, this does not diminish the direct effect of $\kappa$ 
on the posterior mean. At $\kappa = 1$ and $n = 10$ per arm, the 
prior contributes roughly 9\% of the posterior mean, and the 
difference in go rates between priors was modest (0.068 under the 
strong effect). At $\kappa = 10$, the prior weight rises to 50\%, 
and the go rate difference reaches 0.56, with false positive rates 
remaining below 0.01. This identifies $\kappa$ as a clinically 
interpretable design parameter that directly governs how strongly 
prior beliefs influence the trial decision when sample sizes are 
small.

Regarding the role of correlation, the true endpoint correlation 
$\rho_{\mathrm{true}}$ meaningfully increased power, with go 
rates approximately 5 to 10 percentage points higher under 
$\rho_{\mathrm{true}} = 0.6$ compared to independence across all 
configurations. In contrast, the prior correlation $\rho_0$ had 
negligible impact; the maximum absolute difference in go rates 
between $\rho_0 = 0$ and $\rho_0 = 0.6$ was 0.017. This 
asymmetry arises because $\rho_0$ enters the posterior only 
through the prior scale matrix $\Lambda = \lambda_0\Lambda_0$, whose contribution to the posterior trace is bounded. It is negligible under the default $\nu_0 = 2$, constraining $\lambda_0$ to be small 
(Appendix~\ref{app:rho0_insensitivity}). This insensitivity has a 
practical advantage; the practitioner need not estimate or 
specify $\rho_0$ at the design stage, as any value in $(-1, 1)$ 
yields essentially identical operating characteristics.

The telitacicept lupus trial application validated these findings in a real clinical setting. At $\kappa = 1$, the difference in go rates between skeptical and enthusiastic priors was at most 0.013, consistent with the low 
signal-to-noise ratio of the SELENA-SLEDAI endpoint ($\text{SNR} \approx 0.20$). Under these conditions, the likelihood dominates the posterior regardless of prior specification, and the go rate is driven primarily by sample size rather than prior choice. At $\kappa = 10$, 
however, meaningful differences appeared. The enthusiastic go rate at 
the total sample size of 40 was approximately three times the skeptical rate (0.282 versus 0.093). This difference disappeared with increasing sample size, falling to 0.038 at the total sample size of about 300, consistent with the theoretical prediction that larger samples dilute the prior's contribution. This motivating example thus demonstrates that the framework is sensitive to the interaction between $\kappa$, sample size, and the endpoint noise structure, and that prior precision matters most in the small-sample setting where Go/No-Go decisions are typically made.

This model has limitations within this model. First, the analysis model assumes bivariate normality throughout, while the motivating example endpoint was generate from a lognormal distribution to better match the non-negative observed responder rate. The impact of this mismatch between 
the data-generating mechanism and the assumed normal likelihood on posterior calibration deserves further study, particularly for endpoints with heavy skewness. Second, the prior scale matrix $\Lambda_0$ was specified with a unit diagonal, which treats both endpoints as operating on the same scale. In the lupus application, SELENA-SLEDAI has a standard deviation much higher than PGA, so the prior exerts differential influence across endpoints. Third, the correlation between endpoints was estimated indirectly from aggregate statistics under an independence assumption for the BILAG component of the SRI-4 composite, introducing uncertainty that was not propagated through the simulation. 

In summary, the proposed framework provides a principled and interpretable 
approach to prior specification for bivariate Bayesian Go/No-Go decisions. 
The calibration function guarantees existence and uniqueness of the scale 
parameter for any target probability, and the default recommendation of 
$\nu_0 = 2$ simplifies prior parameters to the choice of $\kappa$ and the prior center. The identification of $\kappa$ as the primary driver of prior sensitivity, together with the practical difficulty of specifying the 
optimal $\kappa$ before the trial, provides the central motivation for future work, where $\kappa$ is treated as an adaptive parameter updated at interim to reflect the consistency between prior beliefs and accumulating 
data.

\section{Software}
All analyses were performed in R (version 4.4.0). Multivariate 
$t$ distribution and probabilities were computed using the \texttt{mvtnorm} package~\cite{mvtnorm}. 

\newpage
\bibliographystyle{unsrt}
\bibliography{references}

\newpage
\section{Supplementary Material}
\subsection{Derivation Proof}
\label{sup: proof}
\paragraph{Marginal prior derivation.}
To prove equation \ref{eq:delta distribution}, 
given that the pdf of inverse wishart distribution is approximate to:
\[
  p(\Sigma) 
  =
  \frac{
  \lvert \lambda_0\Lambda_0 \rvert^{\nu_0/2}\lvert \Sigma \rvert^{-(\nu_0 + d + 1)/2}
  \exp\!\Bigl[-\tfrac12 \mathrm{tr}\bigl(\lambda_0\Lambda_0\,\Sigma^{-1}\bigr)\Bigr]}{2^{\nu_0*d/2}\Gamma(\nu_0/2)},
  \quad (d=2).
\]

\begin{align*}
p(\bm{\mu}_A) 
    &= \int p(\bm{\mu}_A \mid \Sigma)\, p(\Sigma)\, d\Sigma \\[6pt]
    &\propto \int (2\pi)^{-\tfrac{1}{2}}\, |\Sigma|^{-\tfrac{1}{2}}
    \exp\!\left(-\tfrac{\kappa}{2} (\bm{\mu}_A - \bm{\mu}_{A0})^\top 
        \Sigma^{-1} (\bm{\mu}_A - \bm{\mu}_{A0})\right) \times 
    \frac{|\lambda_0 \Lambda_0|^{\tfrac{\nu_0}{2}}\, |\Sigma|^{-\tfrac{\nu_0 + d + 1}{2}}
      \exp\!\bigl[-\tfrac{1}{2} \mathrm{tr}(\lambda_0 \Lambda_0 \Sigma^{-1})\bigr]}
      {2^{\tfrac{\nu_0 d}{2}}\, \Gamma(\tfrac{\nu_0}{2})} 
      \, d\Sigma \\[6pt]
    &\propto \int |\Sigma|^{-\tfrac{\nu_0 + d + 2}{2}}
    \exp\!\left(-\tfrac{1}{2} \mathrm{tr}\!\Bigl(\lambda_0 \Lambda_0 \Sigma^{-1}\Bigr)
        -\tfrac{\kappa}{2} (\bm{\mu}_A - \bm{\mu}_{A0})^\top 
        \Sigma^{-1} (\bm{\mu}_A - \bm{\mu}_{A0})\right) d\Sigma \\[6pt]
    &\propto \int |\Sigma|^{-\tfrac{\nu_0 + d + 2}{2}}
    \exp\!\left(-\tfrac{1}{2} \mathrm{tr}\!\Bigl(\lambda_0 \Lambda_0 
        + \kappa (\bm{\mu}_A - \bm{\mu}_{A0})(\bm{\mu}_A - \bm{\mu}_{A0})^\top\Bigr)\Sigma^{-1}\right) d\Sigma \\[6pt]
    &\text{since the integrand is the kernel of } 
    \mathrm{InvWishart}\!\left(\nu_0+1,\;
        \lambda_0 \Lambda_0 + \kappa (\bm{\mu}_A - \bm{\mu}_{A0})(\bm{\mu}_A - \bm{\mu}_{A0})^\top\right), \\[6pt]
    &\propto 
      \frac{2^{\tfrac{(\nu_0+1)d}{2}}\,\Gamma\!\left(\tfrac{\nu_0+1}{2}\right)}
      {\Bigl|\lambda_0 \Lambda_0 + \kappa(\bm{\mu}_A - \bm{\mu}_{A0})(\bm{\mu}_A - \bm{\mu}_{A0})^\top\Bigr|^{\tfrac{\nu_0+1}{2}}}.
\end{align*}

\noindent
Applying the matrix determinant lemma, 
\[
\Bigl|\lambda_0 \Lambda_0 + \kappa(\bm{\mu}_A - \bm{\mu}_{A0})(\bm{\mu}_A - \bm{\mu}_{A0})^\top\Bigr|
= |\lambda_0 \Lambda_0|\left(1 + \tfrac{\kappa}{\lambda_0}(\bm{\mu}_A - \bm{\mu}_{A0})^\top 
   (\Lambda_0)^{-1} (\bm{\mu}_A - \bm{\mu}_{A0})\right).
\]

\noindent
Then, multiplying the normalized constant, Therefore,
\begin{align}
p(\bm{\mu}_A)
&=
   \frac{\Gamma\!\left(\tfrac{\nu_0+1}{2}\right)}
        {\Gamma\!\left(\tfrac{\nu_0}{2}\right)} \,
   \frac{1}{\Bigl[1 + \tfrac{\kappa}{\lambda_0}(\bm{\mu}_A - \bm{\mu}_{A0})^\top 
      (\Lambda_0)^{-1} (\bm{\mu}_A - \bm{\mu}_{A0})\Bigr]^{\tfrac{\nu_0+1}{2}}}.
   \label{eq:marginal mean prior A scaled}
\end{align}

\noindent
Up to the normalizing constant, this is the kernel of a multivariate $t$ distribution:
\[
\bm{\mu}_A \;\sim\; t_{\nu_0+1-d}\!\left(\bm{\mu}_{A0},\; 
    \tfrac{\lambda_0}{\kappa(\nu_0+1-d)} \Lambda_0 \right).
\]

Similarly, 
\[
\bm{\mu}_B \;\sim\; t_{\nu_0+1-d}\!\left(\bm{\mu}_{B0},\; 
    \tfrac{\lambda_0}{\kappa(\nu_0+1-d)} \Lambda_0 \right).
\]

Now, let's define $\bm{\delta} = \bm{\mu}_A - \bm{\mu}_B$, 
\[
  \bm{\delta} \;\big|\; \Sigma
  \;\sim\; \mathcal{N}\!\Bigl(\mathbf{\delta}_{0},\;\tfrac{2}{\kappa}\Sigma\Bigr), \delta_0 =\bm{\mu}_{A0} - \bm{\mu}_{B0} \] 
\begin{align}
p(\bm{\delta})
    &= \int p(\bm{\delta} \mid \Sigma)\, p(\Sigma)\, d\Sigma \notag \\[6pt] 
    &\propto \int (2\pi)^{-\tfrac{d}{2}} 
        \Bigl|\tfrac{2}{\kappa}\Sigma\Bigr|^{-\tfrac{1}{2}}
        \exp\!\left(-\tfrac{1}{2} 
            (\bm{\delta}-\bm{\delta}_0)^\top 
            \Bigl(\tfrac{2}{\kappa}\Sigma\Bigr)^{-1}
            (\bm{\delta}-\bm{\delta}_0)\right) \notag\\
    &\qquad\quad \times 
        \frac{|\lambda_0 \Lambda_0|^{\tfrac{\nu_0}{2}} 
              |\Sigma|^{-\tfrac{\nu_0+d+1}{2}}
              \exp\!\left[-\tfrac{1}{2}\mathrm{tr}(\lambda_0 \Lambda_0 \Sigma^{-1})\right]}
             {2^{\tfrac{\nu_0 d}{2}} \Gamma_d(\tfrac{\nu_0}{2})} 
        \, d\Sigma \notag\\
    &\propto 
    \frac{2^{\tfrac{(\nu_0+1)d}{2}}\,\Gamma\!\bigl(\tfrac{\nu_0+1}{2}\bigr)}
     {\Bigl|\lambda_0 \Lambda_0 + \tfrac{\kappa}{2}(\bm{\delta}-\bm{\delta}_0)(\bm{\delta}-\bm{\delta}_0)^\top\Bigr|^{\tfrac{\nu_0+1}{2}}} \notag\\
     &=
   \frac{\Gamma\!\left(\tfrac{\nu_0+1}{2}\right)}
        {\Gamma\!\left(\tfrac{\nu_0}{2}\right)} \,
   \frac{1}{\Bigl[1 + \tfrac{\kappa}{2\lambda_0}
       (\bm{\delta}-\bm{\delta}_0)^\top \Lambda_0^{-1}(\bm{\delta}-\bm{\delta}_0)\Bigr]^{\tfrac{\nu_0+1}{2}}}.
   \label{eq:marginal prior delta}
\end{align}

\noindent
Thus,
\[
\bm{\delta} \;\sim\; t_{\nu_0+1-d}\!\left(
    \bm{\delta}_0,\;
    \tfrac{2\lambda_0}{\kappa(\nu_0+1-d)} \Lambda_0
\right).
\]
 
\paragraph{Posterior derivation.}

The posterior distribution is proportional to the product of the likelihood and the prior,
\[
p(\boldsymbol{\mu}_A, \boldsymbol{\mu}_B, \Sigma \mid \mathbf{Y})
\;\propto\;
L(\boldsymbol{\mu}_A, \boldsymbol{\mu}_B, \Sigma \mid \mathbf{Y})
\times
p(\boldsymbol{\mu}_A, \boldsymbol{\mu}_B, \Sigma).
\]

\noindent
For the arm-specific means, the conjugate update implies
\[
\boldsymbol{\mu}_A \mid \Sigma, \mathbf{Y}
\sim \mathcal{N}\!\left(\boldsymbol{\mu}_A^*,\, \frac{1}{\kappa_A^*}\Sigma\right),
\]
where
\[
\boldsymbol{\mu}_A^* = 
\frac{\kappa_{A0} \boldsymbol{\mu}_{A0} + n_A \bar{\mathbf{Y}}_A}{\kappa_{A0} + n_A},
\qquad
\kappa_A^* = \kappa_{A0} + n_A,
\]
and analogously for arm $B$.

\medskip
\noindent
For the covariance matrix $\Sigma$, we first write the likelihood conditional on 
$\boldsymbol{\mu}_A, \boldsymbol{\mu}_B$ as
\begin{align}
p(\mathbf{Y}_A, \mathbf{Y}_B \mid \Sigma, \boldsymbol{\mu}_A, \boldsymbol{\mu}_B)
&\propto
\prod_{i=1}^{n_A}
|\Sigma|^{-1/2}
\exp\!\left\{
-\frac{1}{2}
(\mathbf{Y}_{Ai} - \boldsymbol{\mu}_A)^\top
\Sigma^{-1}
(\mathbf{Y}_{Ai} - \boldsymbol{\mu}_A)
\right\}
\notag\\
&\quad\times
\prod_{j=1}^{n_B}
|\Sigma|^{-1/2}
\exp\!\left\{
-\frac{1}{2}
(\mathbf{Y}_{Bj} - \boldsymbol{\mu}_B)^\top
\Sigma^{-1}
(\mathbf{Y}_{Bj} - \boldsymbol{\mu}_B)
\right\} \notag\\
&\propto
|\Sigma|^{-\frac{n_A}{2}}
\exp\!\left(
-\frac{1}{2}
\sum_{i=1}^{n_A}
\mathrm{tr}\!\left[
\Sigma^{-1}
(\mathbf{Y}_{Ai} - \boldsymbol{\mu}_A)
(\mathbf{Y}_{Ai} - \boldsymbol{\mu}_A)^\top
\right]
\right)
\notag\\
&\quad\times
|\Sigma|^{-\frac{n_B}{2}}
\exp\!\left(
-\frac{1}{2}
\sum_{j=1}^{n_B}
\mathrm{tr}\!\left[
\Sigma^{-1}
(\mathbf{Y}_{Bj} - \boldsymbol{\mu}_B)
(\mathbf{Y}_{Bj} - \boldsymbol{\mu}_B)^\top
\right]
\right).\label{eq:likelihood-trace}
\end{align}

Define the scatter matrices
\[
S_A 
= \sum_{i=1}^{n_A}
(\mathbf{Y}_{Ai} - \boldsymbol{\mu}_A)
(\mathbf{Y}_{Ai} - \boldsymbol{\mu}_A)^\top,
\qquad
S_B 
= \sum_{j=1}^{n_B}
(\mathbf{Y}_{Bj} - \boldsymbol{\mu}_B)
(\mathbf{Y}_{Bj} - \boldsymbol{\mu}_B)^\top.
\]
Then \eqref{eq:likelihood-trace} can be written as 
\begin{align}
p(\mathbf{Y}_A, \mathbf{Y}_B \mid \Sigma, \boldsymbol{\mu}_A, \boldsymbol{\mu}_B)
&\propto
|\Sigma|^{-\frac{n_A}{2}}
\exp\!\left(
-\frac{1}{2}\mathrm{tr}\bigl(\Sigma^{-1} S_A\bigr)
\right)
\times
|\Sigma|^{-\frac{n_B}{2}}
\exp\!\left(
-\frac{1}{2}\mathrm{tr}\bigl(\Sigma^{-1} S_B\bigr)
\right) \notag\\
&\propto
|\Sigma|^{-\frac{n_A+n_B}{2}}
\exp\!\left(
-\frac{1}{2}\mathrm{tr}\bigl[\Sigma^{-1}(S_A + S_B)\bigr]
\right).
\label{eq:joint-likelihood-SA-SB}
\end{align}

The Inverse-Wishart prior is
\[
p(\Sigma)
\propto
|\Sigma|^{-(\nu_0 + d + 1)/2}
\exp\!\left(
-\frac{1}{2}\mathrm{tr}\bigl(\Lambda\,\Sigma^{-1}\bigr)
\right),
\]
where $d$ is the dimension ($d=2$).  

Combining the likelihood \eqref{eq:joint-likelihood-SA-SB} and the prior, the conditional posterior of $\Sigma$ given 
$\mathbf{Y}_A, \mathbf{Y}_B, \boldsymbol{\mu}_A, \boldsymbol{\mu}_B$ is

\begin{align}
p(\Sigma \mid \mathbf{Y}_A, \mathbf{Y}_B, \boldsymbol{\mu}_A, \boldsymbol{\mu}_B)
&\propto
|\Sigma|^{-\frac{n_A+n_B}{2}}
\exp\!\left(
-\frac{1}{2}\mathrm{tr}\bigl[\Sigma^{-1}(S_A + S_B)\bigr]
\right)
\times
|\Sigma|^{-(\nu_0 + d + 1)/2}
\exp\!\left(
-\frac{1}{2}\mathrm{tr}\bigl(\Lambda\,\Sigma^{-1}\bigr)
\right) \notag\\
&\propto
|\Sigma|^{-\frac{\nu_0 + n_A + n_B + d + 1}{2}}
\exp\!\left(
-\frac{1}{2}\mathrm{tr}\bigl[\Sigma^{-1}(\Lambda + S_A + S_B)\bigr]
\right).
\end{align}

Recognizing this kernel as that of an Inverse-Wishart distribution, we obtain
\[
\Sigma \mid \mathbf{Y}_A, \mathbf{Y}_B, \boldsymbol{\mu}_A, \boldsymbol{\mu}_B
\sim
\mathrm{InvWishart}\bigl(\nu_0 + n_A + n_B,\; \Lambda + S_A + S_B\bigr).
\]

\begin{align}
\mathbb{E}[\Sigma \mid D_A, D_B] = \frac{S_A + S_B + \Lambda}{n_A + n_B + \nu_0 - d - 1}
\end{align}

\subsection{Insensitivity to prior correlation}
\label{app:rho0_insensitivity}

The prior correlation $\rho_0$ enters the posterior of $\Sigma$ 
only through the scale matrix $\Lambda + S_A + S_B$, where 
$\Lambda = \lambda_0 \Lambda_0$. We show that under the default 
$\nu_0 = 2$, the contribution of $\rho_0$ to the posterior is 
negligible.

\begin{proof} 
By linearity of the trace, from Equation 15, 
\begin{align*}
\mathrm{tr}\!\left[\Sigma^{-1}(\Lambda + S_A + S_B)\right]
= \mathrm{tr}\!\left[\Sigma^{-1}\lambda_0\Lambda_0\right]
+ \mathrm{tr}\!\left[\Sigma^{-1}(S_A + S_B)\right].
\end{align*}

Evaluate the first term. With $\sigma_1 = \sigma_2 = 1$,
\begin{align*}
\Sigma^{-1} = \frac{1}{1-\rho^2}
\begin{pmatrix} 1 & -\rho \\ -\rho & 1 \end{pmatrix},
\qquad
\Lambda_0 = 
\begin{pmatrix} 1 & \rho_0 \\ \rho_0 & 1 \end{pmatrix}.
\end{align*}
Then,
\begin{align*}
\Sigma^{-1}\lambda_0\Lambda_0
&= \frac{\lambda_0}{1-\rho^2}
\begin{pmatrix} 1 & -\rho \\ -\rho & 1 \end{pmatrix}
\begin{pmatrix} 1 & \rho_0 \\ \rho_0 & 1 \end{pmatrix}
= \frac{\lambda_0}{1-\rho^2}
\begin{pmatrix} 
1-\rho\rho_0 & \rho_0-\rho \\ 
-\rho+\rho_0 & 1-\rho\rho_0 
\end{pmatrix}.
\end{align*}

Taking the trace,
\begin{align*}
\mathrm{tr}\!\left[\Sigma^{-1}\lambda_0\Lambda_0\right]
= \frac{2\lambda_0(1-\rho\,\rho_0)}{1-\rho^2}.
\label{eq:prior_trace}
\end{align*}

The second term $\mathrm{tr}[\Sigma^{-1}(S_A + S_B)]$ depends 
only on the data through the scatter matrices and contains no 
prior hyperparameters. Therefore, $\rho_0$ can only influence 
the posterior through the first term 
$\mathrm{tr}[\Sigma^{-1}\lambda_0\Lambda_0]$.

The change in the prior trace contribution when $\rho_0$ shifts 
from $0$ to some value $\rho_0'$ is
\begin{align*}
\Delta_{\rho_0} 
= \left|\frac{2\lambda_0(1-\rho\,\rho_0')}{1-\rho^2} 
  - \frac{2\lambda_0}{1-\rho^2}\right|
= \frac{2\lambda_0\,|\rho|\,|\rho_0'|}{1-\rho^2}.
\end{align*}

For the second term, with $S_A + S_B \approx (n_A+n_B)\hat{\Sigma}$ 
where $\hat{\Sigma} \approx \Sigma$,
\begin{align*}
\mathrm{tr}\!\left[\Sigma^{-1}(S_A+S_B)\right]
\approx (n_A+n_B)\,\mathrm{tr}(I_d) 
= d(n_A+n_B).
\end{align*}

The ratio of the $\rho_0$-induced change to the data contribution 
is therefore
\begin{align}
\frac{\Delta_{\rho_0}}{d(n_A+n_B)}
= \frac{2\lambda_0\,|\rho|\,|\rho_0'|}{(1-\rho^2)\,d(n_A+n_B)}
= O\!\left(\frac{\lambda_0}{n}\right).
\end{align}

Under the default $\nu_0 = 2$, the calibrated $\lambda_0$ is 
small (e.g., $\lambda_0 = 0.02$ at $\kappa = 1$, and 
$\lambda_0 = 0.194$ at $\kappa = 10$). 

Thus, even at the smallest sample size considered ($n = 10$ per 
arm) and the largest prior precision ($\kappa = 10$), the ratio 
remains below $1\%$ for moderate correlations 
($|\rho| \leq 0.6$). The factor $1/(1-\rho^2)$ in the prior 
trace term means that very high true correlations ($|\rho| > 0.9$) could amplify the prior's contribution, but 
the ratio remains bounded by $O(\lambda_0/n)$ and does not 
exceed $3\%$ even at $|\rho| = 0.9$ under the default 
$\nu_0 = 2$. Crucially, $\rho_0$ enters only in the numerator 
through $|\rho_0'| \leq 1$, so changing $\rho_0$ from $0$ to 
any value in $(-1, 1)$ cannot increase the order of the ratio 
beyond $O(\lambda_0/n)$.
\end{proof}

\subsection{Proposition proof}
\label{app: proposition proof}

\begin{lemma}[Continuity]\label{lem:continuity}
For fixed $\nu_0 > d-1$, $\kappa > 0$, and 
$\bm{\delta}_0 \neq \boldsymbol{\tau}$, 
$f(\lambda_0)$ is continuous in $\lambda_0$ on $(0, \infty)$.
\end{lemma}

\begin{proof}
Each component of $\mathbf{z}(\lambda_0)$ is proportional to 
$\lambda_0^{-1/2}$, which is continuous on $(0, \infty)$. The 
bivariate $t$ survival function $\bar{F}_{\nu_0+1-d}(\mathbf{z}; R)$ 
is continuous in $\mathbf{z}$ because it is defined by integration 
of a continuous density over a region whose boundary moves 
continuously with $\mathbf{z}$. The composition of continuous 
functions is continuous.
\end{proof}

\begin{lemma}[Monotonicity]\label{lem:monotone}
Under the same conditions, $f(\lambda_0)$ is strictly monotone in 
$\lambda_0$ on $(0, \infty)$.
\end{lemma}

\begin{proof}
As $\lambda_0$ increases, each component of 
$\mathbf{z}(\lambda_0)$ moves toward zero, since 
$|z_j(\lambda_0)| \propto \lambda_0^{-1/2}$. For the skeptical 
prior, $\bm{\delta}_0$ lies below $\boldsymbol{\tau}$ 
componentwise, so $\boldsymbol{\tau} - \bm{\delta}_0 > \mathbf{0}$ 
and $\mathbf{z}(\lambda_0) > \mathbf{0}$. As $\lambda_0$ 
increases, $\mathbf{z}(\lambda_0)$ decreases toward $\mathbf{0}$, 
and the upper orthant probability 
$\bar{F}_{\nu_0+1-d}(\mathbf{z}; R)$ increases (since less of the 
distribution is excluded). Hence $f$ is strictly increasing in 
$\lambda_0$.

For the enthusiastic prior, $\bm{\delta}_0$ lies above 
$\boldsymbol{\tau}$, so by the analogous argument applied to the 
lower orthant, $f$ is also strictly increasing in $\lambda_0$.

In both cases, the strict monotonicity follows from the fact that 
the bivariate $t$ density is strictly positive everywhere, so any 
shift in $\mathbf{z}$ produces a strict change in the orthant 
probability.
\end{proof}

\begin{lemma}[Boundary limits]\label{lem:limits}
Under the same conditions,
\[
\lim_{\lambda_0 \to 0} f(\lambda_0) = 0, 
\qquad
\lim_{\lambda_0 \to \infty} f(\lambda_0) = L,
\]
where $L = P(T_1 > 0,\, T_2 > 0)$ is the positive orthant 
probability of a centered bivariate $t$ distribution with 
correlation matrix $\Lambda_0$ and $\nu_0+1-d$ degrees of freedom.
\end{lemma}

\begin{proof}
\emph{(i) As $\lambda_0 \to 0$:} each component of 
$\mathbf{z}(\lambda_0)$ diverges to $+\infty$ (skeptical case) or 
$-\infty$ (enthusiastic case). In the skeptical case, the upper 
orthant probability $P(\mathbf{T} > \mathbf{z})$ vanishes as 
$\mathbf{z} \to +\boldsymbol{\infty}$. In the enthusiastic case, 
the lower orthant probability $P(\mathbf{T} < \mathbf{z})$ 
vanishes as $\mathbf{z} \to -\boldsymbol{\infty}$. Hence 
$f(\lambda_0) \to 0$ in both cases.

\emph{(ii) As $\lambda_0 \to \infty$:} 
$\mathbf{z}(\lambda_0) \to \mathbf{0}$, so
\[
f(\lambda_0) \to P(\mathbf{T} > \mathbf{0}) 
= P(T_1 > 0,\, T_2 > 0) = L.
\]
When $\rho_0 = 0$, the bivariate $t$ orthant probability equals 
$1/4$ by symmetry. More generally, for large $\nu_0$ the 
$t$ distribution is well approximated by a bivariate normal, 
yielding~\cite{abrahamson1964orthant}
\[
L \approx \frac{1}{4} + \frac{1}{2\pi}\arcsin(\rho_0).
\]
Since $L \in (0,1)$ and $f$ is continuous and strictly increasing 
from $0$ to $L$, the intermediate value theorem guarantees a unique 
$\lambda_0^* > 0$ satisfying $f(\lambda_0^*) = \alpha$ for any 
$\alpha \in (0, L)$. \qed

\noindent\textit{Proof of Proposition~\ref{prop:existence}.} 
Lemma~\ref{lem:continuity} establishes that $f(\lambda_0)$ is 
continuous on $(0, \infty)$, Lemma~\ref{lem:monotone} that it is 
strictly monotone, and Lemma~\ref{lem:limits} that it maps onto 
$(0, L)$. By the intermediate value theorem, every 
$\alpha \in (0, L)$ is attained, and strict monotonicity guarantees 
the solution $\lambda_0^*$ is unique. Since these properties hold 
for any fixed $\kappa > 0$, the existence and uniqueness of 
$\lambda_0^*(\kappa)$ extends to all $\kappa$. \qed
\end{proof}

\setcounter{table}{0}
\renewcommand{\thetable}{S\arabic{table}}

\subsection{TABLES}

\begin{table}[ht]
\centering
\caption{Prior specifications for the eight prior distributions. For each prior, 
         $\lambda_0$ is calibrated via calibration function where the 
         tail probability at the specified threshold equals the target 
         probability, reported separately for each value of $\nu_0$. 
         All priors use $\kappa_{A0} = \kappa_{B0} = 1$ and $\rho_0 = 0$.}
\label{tab:s1}
\resizebox{\textwidth}{!}{%
\begin{tabular}{lcccccccc}
\toprule
Prior & Center & Threshold  
      & Target Prob. & Region 
      & \multicolumn{4}{c}{$\lambda_0$} \\
\cmidrule(lr){6-9}
& & & & & $\nu_0 = 2$ & $\nu_0 = 4$ & $\nu_0 = 10$ & $\nu_0 = 50$ \\
\midrule
\multicolumn{9}{l}{\textit{Skeptical priors}} \\
Skeptical            & $(0.0,\ 0.0)$   & $(0.5,\ 0.5)$   & 0.05 & Upper &0.02 &0.2 &0.84 &5.16 \\
Extreme skeptical    & $(0.0,\ 0.0)$   & $(0.5,\ 0.5)$   & 0.01 & Upper &0.01 &0.04 &0.26 &1.77 \\
Extreme skeptical v2 & $(-1.0,\ -1.0)$ & $(-0.5,\ -0.5)$ & 0.05 & Upper &0.02 &0.2 &0.84 &5.16 \\
Extreme skeptical v3 & $(-1.0,\ -1.0)$ & $(-0.5,\ -0.5)$ & 0.01 & Upper &0.01 &0.04 &0.26 &1.77 \\
\addlinespace
\multicolumn{9}{l}{\textit{Enthusiastic priors}} \\
Enthusiastic            & $(0.5,\ 0.5)$ & $(0.0,\ 0.0)$ & 0.05 & Lower &0.02 &0.2 &0.84 &5.16 \\
Extreme enthusiastic    & $(0.5,\ 0.5)$ & $(0.0,\ 0.0)$ & 0.01 & Lower &0.01 &0.04 &0.26 &1.77 \\
Extreme enthusiastic v2 & $(1.5,\ 1.5)$ & $(1.0,\ 1.0)$ & 0.05 & Lower &0.02 &0.2 &0.84 &5.16 \\
Extreme enthusiastic v3 & $(1.5,\ 1.5)$ & $(1.0,\ 1.0)$ & 0.01 & Lower &0.01 &0.04 &0.26 &1.77 \\
\bottomrule
\end{tabular}%
}
\begin{tablenotes}
\small
\item Prior center denotes the location parameter $(\delta_1, \delta_2)$ 
      of the marginal multivariate $t$ prior on the treatment effect 
      $\delta = \mu_A - \mu_B$. Threshold denotes the boundary at which 
      the tail probability is calibrated. Upper tail: 
      $P(\delta > \text{threshold}) = \text{target prob.}$; Lower tail: 
      $P(\delta < \text{threshold}) = \text{target prob.}$ 
\end{tablenotes}
\end{table}

\begin{table}[ht]
\centering
\caption{Operating characteristics of the Bayesian Go/No-Go decision rule 
         under the both-effects scenario with $n = 10$ per arm, 
         $\rho_0 = 0$, $\rho_{\text{true}} = 0$, and 
         10{,}000 simulation replications. Three true effect sizes are 
         considered: null ($\boldsymbol{\delta} = \mathbf{0}$), moderate 
         ($\boldsymbol{\delta} = 0.5$), and strong ($\boldsymbol{\delta} = 1$), 
         where $\boldsymbol{\delta} = (\delta_1, \delta_2)$.
         Go rate is the proportion of simulations in which 
         $P(\delta_1 > 0, \delta_2 > 0 \mid \text{data}) > 0.95$. 
         Under the null scenario, go rate corresponds to the empirical 
         false positive rate. Prior specifications are detailed in 
         Table~\ref{tab:priors}.}
\label{tab:s2}
\begin{tabular}{llcccccc}
\toprule
& & \multicolumn{3}{c}{Go Rate} 
  & \multicolumn{3}{c}{Avg.\ Posterior Prob.} \\
\cmidrule(lr){3-5} \cmidrule(lr){6-8}
$\nu_0$ & Prior & Null & Moderate & Strong 
               & Null & Moderate & Strong \\
\midrule
\multirow{2}{*}{2}
  & Skeptical    & 0.0021 & 0.0658 & 0.4417 & 0.2502 & 0.6244 & 0.8835 \\
  & Enthusiastic & 0.0027 & 0.0902 & 0.5097 & 0.2848 & 0.6609 & 0.9018 \\
\addlinespace
\multirow{2}{*}{4}
  & Skeptical    & 0.0024 & 0.0801 & 0.4824 & 0.2502 & 0.6320 & 0.8910 \\
  & Enthusiastic & 0.0042 & 0.1081 & 0.5487 & 0.2854 & 0.6689 & 0.9086 \\
\addlinespace
\multirow{2}{*}{10}
  & Skeptical    & 0.0048 & 0.1193 & 0.5692 & 0.2503 & 0.6495 & 0.9071 \\
  & Enthusiastic & 0.0076 & 0.1541 & 0.6304 & 0.2869 & 0.6872 & 0.9232 \\
\addlinespace
\multirow{2}{*}{50}
  & Skeptical    & 0.0223 & 0.2806 & 0.7764 & 0.2506 & 0.6970 & 0.9428 \\
  & Enthusiastic & 0.0305 & 0.3333 & 0.8145 & 0.2911 & 0.7360 & 0.9546 \\
\bottomrule
\end{tabular}
\begin{tablenotes}
\small
\item 
      For each $\nu_0$, $\lambda_0$ is uniquely determined by the 
      calibration function (Proposition~\ref{prop:existence}).
      All priors use $\kappa_{A0} = \kappa_{B0} = 1$.
\end{tablenotes}
\end{table}

\begin{table}[ht]
\centering
\caption{Go rates for skeptical and enthusiastic priors across 
         $\kappa \in \{1, 5, 10\}$ and sample sizes, under the 
         both-effects scenario, $\nu_0 = 2$, $\rho_0 = 0.6$, 
         $c = 0.95$, $10{,}000$ replications. For each $\kappa$, 
         $\lambda_0$ is recalibrated to maintain the target tail 
         probability of $0.05$.}
\label{tab:kappa_sensitivity_rho06}
\begin{tabular}{llcccccc}
\toprule
& & \multicolumn{3}{c}{$\rho_0 = 0.6,\; \rho_{\mathrm{true}} = 0$} 
  & \multicolumn{3}{c}{$\rho_0 = 0.6,\; \rho_{\mathrm{true}} = 0.6$} \\
\cmidrule(lr){3-5} \cmidrule(lr){6-8}
$\kappa$ & Prior & Null & Mod. & Strong 
                  & Null & Mod. & Strong \\
\midrule
\multicolumn{8}{l}{\textit{$n = 10$ per arm}} \\
\addlinespace
\multirow{2}{*}{1}
  & Skeptical    & 0.0012 & 0.0696 & 0.4548 & 0.0141 & 0.1418 & 0.5327 \\
  & Enthusiastic & 0.0018 & 0.0937 & 0.5229 & 0.0185 & 0.1755 & 0.5939 \\
\addlinespace
\multirow{2}{*}{5}
  & Skeptical    & 0.0001 & 0.0224 & 0.2347 & 0.0050 & 0.0674 & 0.3404 \\
  & Enthusiastic & 0.0035 & 0.1197 & 0.5674 & 0.0253 & 0.2111 & 0.6310 \\
\addlinespace
\multirow{2}{*}{10}
  & Skeptical    & 0.0000 & 0.0032 & 0.0696 & 0.0013 & 0.0252 & 0.1609 \\
  & Enthusiastic & 0.0078 & 0.1786 & 0.6491 & 0.0412 & 0.2837 & 0.6957 \\
\midrule
\multicolumn{8}{l}{\textit{$n = 25$ per arm}} \\
\addlinespace
\multirow{2}{*}{1}
  & Skeptical    & 0.0011 & 0.2577 & 0.9315 & 0.0129 & 0.3654 & 0.9397 \\
  & Enthusiastic & 0.0017 & 0.2920 & 0.9414 & 0.0153 & 0.3982 & 0.9480 \\
\addlinespace
\multirow{2}{*}{5}
  & Skeptical    & 0.0005 & 0.1958 & 0.8987 & 0.0087 & 0.3034 & 0.9131 \\
  & Enthusiastic & 0.0029 & 0.3534 & 0.9580 & 0.0244 & 0.4604 & 0.9631 \\
\addlinespace
\multirow{2}{*}{10}
  & Skeptical    & 0.0001 & 0.1334 & 0.8471 & 0.0050 & 0.2368 & 0.8668 \\
  & Enthusiastic & 0.0056 & 0.4470 & 0.9736 & 0.0360 & 0.5409 & 0.9763 \\
\midrule
\multicolumn{8}{l}{\textit{$n = 50$ per arm}} \\
\addlinespace
\multirow{2}{*}{1}
  & Skeptical    & 0.0011 & 0.6152 & 0.9988 & 0.0119 & 0.6718 & 0.9995 \\
  & Enthusiastic & 0.0019 & 0.6377 & 0.9991 & 0.0136 & 0.6907 & 0.9995 \\
\addlinespace
\multirow{2}{*}{5}
  & Skeptical    & 0.0006 & 0.5793 & 0.9987 & 0.0093 & 0.6411 & 0.9989 \\
  & Enthusiastic & 0.0033 & 0.6984 & 0.9997 & 0.0191 & 0.7394 & 0.9996 \\
\addlinespace
\multirow{2}{*}{10}
  & Skeptical    & 0.0004 & 0.5309 & 0.9982 & 0.0070 & 0.6049 & 0.9983 \\
  & Enthusiastic & 0.0056 & 0.7659 & 0.9998 & 0.0284 & 0.7893 & 0.9998 \\
\bottomrule
\end{tabular}
\begin{tablenotes}
\small
\item $\rho_0$: prior correlation in $\Lambda_0$; 
      $\rho_{\mathrm{true}}$: true correlation used to generate 
      data. Null: $\bm{\delta} = (0,0)$; Moderate (Mod.): 
      $\bm{\delta} = (0.5, 0.5)$; Strong: 
      $\bm{\delta} = (1, 1)$.
\end{tablenotes}
\end{table}
\FloatBarrier

\setcounter{figure}{0}
\renewcommand{\thefigure}{S\arabic{figure}}

\subsection{Figures}
\begin{figure}[ht]
    \centering
    \begin{subfigure}[b]{0.8\textwidth}
        \centering
        \includegraphics[width=\textwidth]{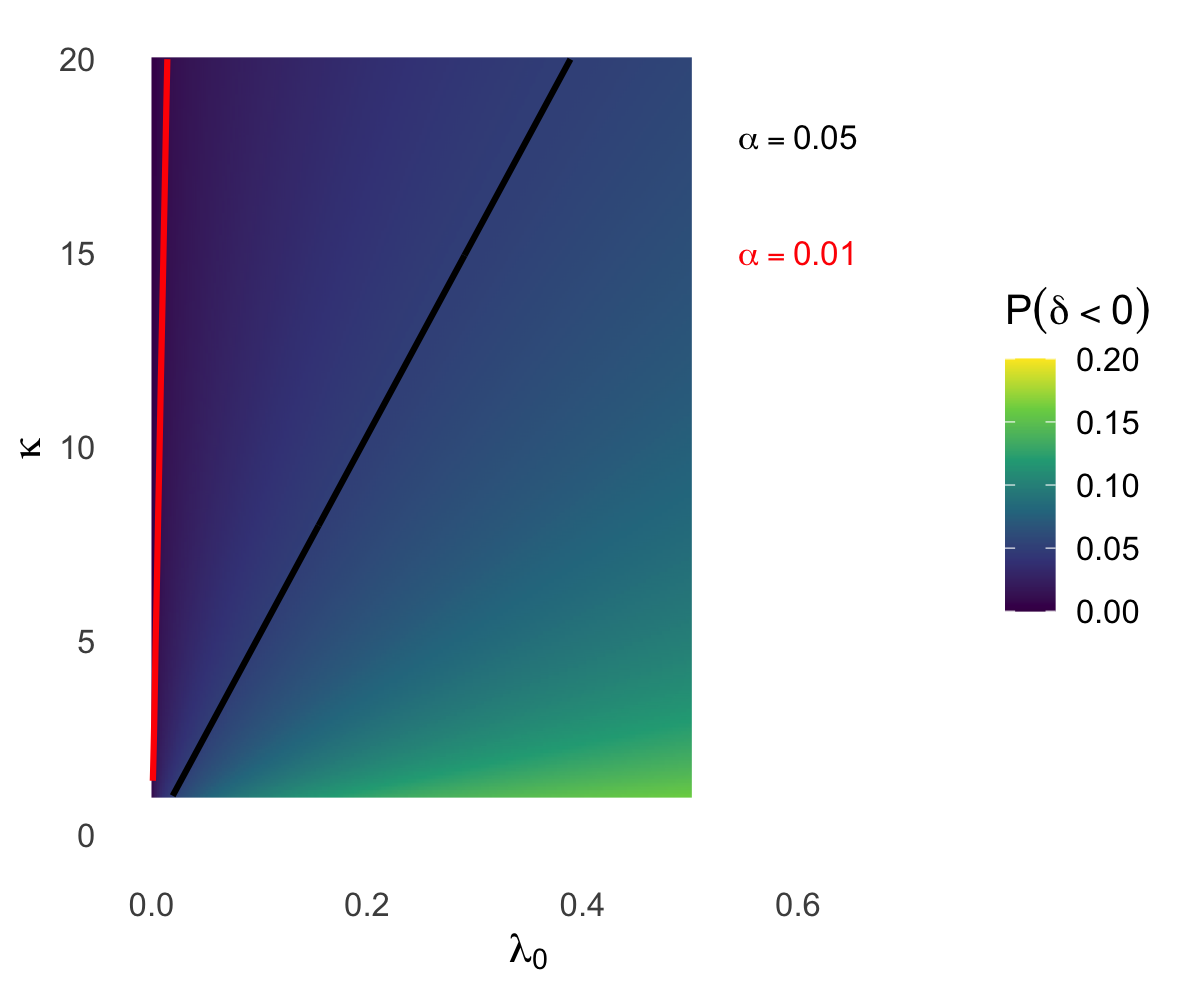}
        \caption{$\rho_0 = 0$}
        \label{fig:heatmap_rho0}
    \end{subfigure}
    \hfill
    \begin{subfigure}[b]{0.8\textwidth}
        \centering
        \includegraphics[width=\textwidth]{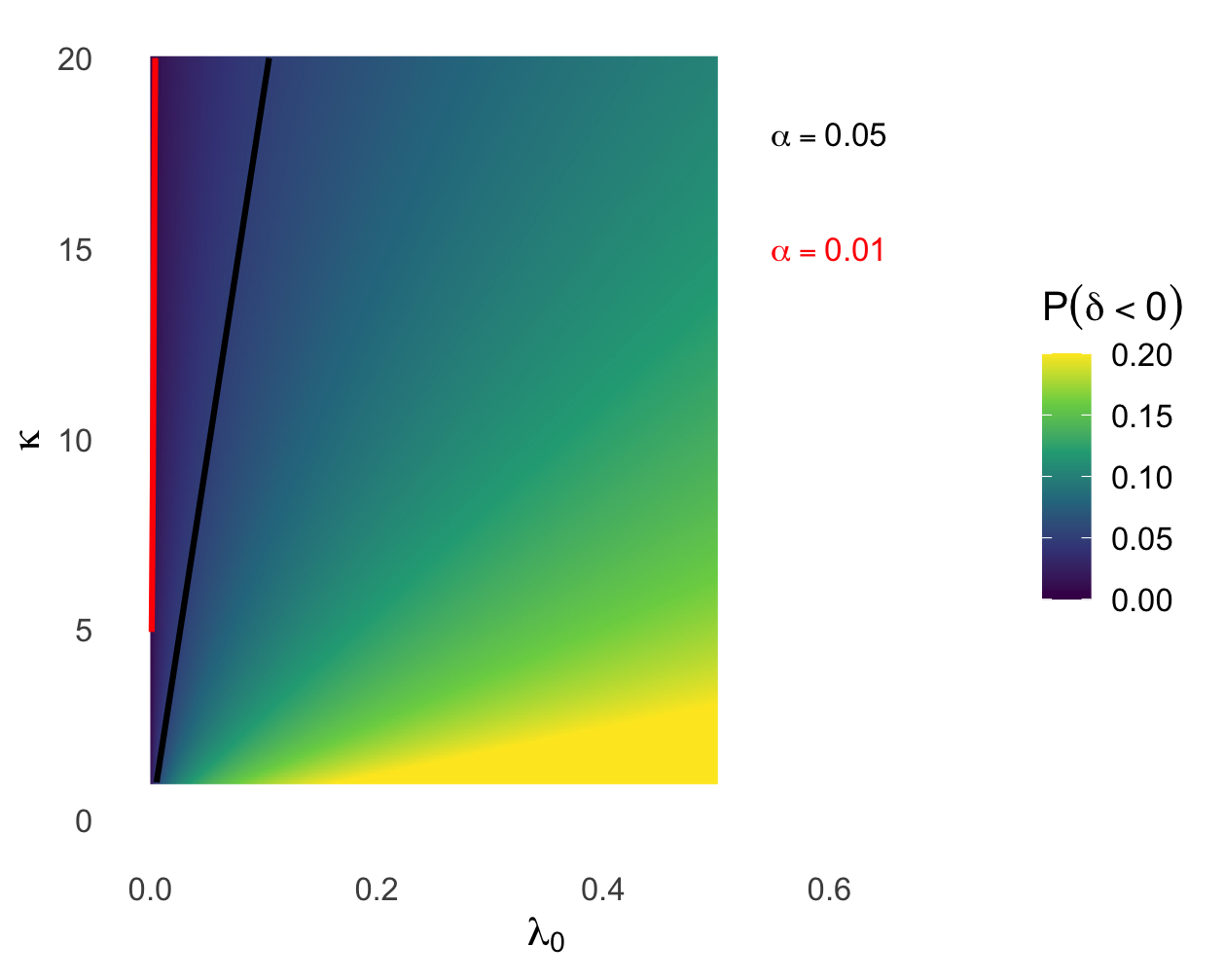}
        \caption{$\rho_0 = 0.6$}
        \label{fig:heatmap_rho06}
    \end{subfigure}
    \caption{Calibration function $P(\bm{\delta} > \boldsymbol{\tau})$ 
             for the enthusiastic prior across $(\kappa, \lambda_0)$ with 
             $\nu_0 = 2$. Contour lines indicate $\alpha = 0.05$ 
             (black) and $\alpha = 0.01$ (red). The contour lines 
             trace the $(\kappa, \lambda_0)$ pairs satisfying the 
             calibration in Proposition~\ref{prop:existence}.}
    \label{fig:prior_enth_heatmap}
\end{figure}

\end{document}